\documentclass[11pt,a4paper]{amsart}
\usepackage{amsthm}
\usepackage{amsfonts}
\usepackage{amsmath}
\usepackage{amssymb}
\usepackage{ifthen}
\usepackage{tikz}
\usepackage{subfigure}
\usepackage{hyperref}
\usetikzlibrary{matrix}

\title{Utility Maximization, Risk Aversion, and Stochastic Dominance}
\author[M. Beiglb\"ock]{Mathias Beiglb\"ock}
\address{Universit\"at Wien, Fakult\"at f\"ur Mathematik, Nordbergstrasse 15 \endgraf A-1090 Wien, Austria}
\email{mathias.beiglboeck@univie.ac.at }
\author[J. Muhle-Karbe]{Johannes Muhle-Karbe}
\address{ETH Z\"urich, Departement Mathematik, R\"amistrasse 101 \endgraf CH-8092, Z\"urich, Switzerland}
\email{johannes.muhle-karbe@math.ethz.ch}
\author[J. Temme]{Johannes Temme}
\address{Universit\"at Wien, Fakult\"at f\"ur Mathematik, Nordbergstrasse 15 \endgraf A-1090 Wien, Austria}
\email{johannes.temme@univie.ac.at}
\date{\today}

\thanks{The first and third author gratefully acknowledge financial support from the Austrian Science Fund (FWF),
under grant P21209 resp.\  P19456. The second author gratefully acknowledges financial support by the National Centre of
Competence in Research ``Financial Valuation and Risk Management'' (NCCR FINRISK),
Project D1 (Mathematical Methods in Financial Risk Management), of the Swiss National Science Foundation.
All authors thank Marcel Nutz and Walter Schachermayer for fruitful discussions, and an anonymous referee for numerous constructive comments.}


  \numberwithin{equation}{section}

  \theoremstyle{definition}
  \newtheorem{definition}{Definition}[section]
  \renewcommand{\thedefinition}{\arabic{section}.\arabic{definition}}
  
  \theoremstyle{definition}
  \newtheorem{remark}[definition]{Remark}

  \theoremstyle{plain}
  \newtheorem{theorem}[definition]{Theorem}

  \theoremstyle{plain}
  \newtheorem{lemma}[definition]{Lemma}

  \theoremstyle{plain}
  \newtheorem{proposition}[definition]{Proposition}

  \theoremstyle{plain}
  \newtheorem{corollary}[definition]{Corollary}

  \newcommand{\BbbR}{\ensuremath{\mathbb {R}}}
  
  \newcommand{\BbbN}{{\mathbb {N}}}
  
  \newcommand{\BbbP}{{\mathbb {P}}}

  \newcommand{\agent}[1]{{\ensuremath{\ifthenelse{\equal{#1}{1}}{L}{\ifthenelse{\equal{#1}{2}}{M}{\ifthenelse{\equal{#1}{RL}}{L}{\ifthenelse{\equal{#1}{RA}}{M}{TBA}}}}}}}
  \newcommand{\tildeexpval}[2][]{\ensuremath{\ifthenelse{\equal{#1}{}}{\tilde{\mathbb{E}}\left[#2\right]}{\tilde{\mathbb{E}}\left[#2|#1\right]}}}

  \newcommand{\Q}{\mathbb Q}
  \renewcommand{\P}{\mathbb P}
  
  \newcommand{\convmon}{\leq_{MC}}
  \newcommand{\conv}{\leq_{C}}
  \newcommand{\mra}{\ll}

  \newcommand{\stochexp}[1]{\ensuremath{\mathcal{E}\left(#1\right)}}
  
  \newcommand{\expval}[2][]{\ensuremath{\ifthenelse{\equal{#1}{}}{\mathbb{E}\left[#2\right]}{\mathbb{E}\left[#2|#1\right]}}}

\renewcommand{\epsilon}{\varepsilon}
\newcommand{\eps}{\varepsilon}
\newcommand{\E}{\mathbb {E}}
\newcommand{\R}{\mathbb R}
\begin{document}

  \begin{abstract}
    Consider an investor trading dynamically to maximize expected
utility from terminal wealth. Our aim is to study the dependence
between her risk aversion and the distribution of the optimal terminal
payoff.
    Economic intuition suggests that high risk aversion leads to a
rather concentrated distribution, whereas lower risk aversion results
in a higher average payoff at the expense of a more widespread
distribution.
    Dybvig and Wang [\textit{J. Econ.\ Theory}, 2011, to appear] find
that this idea can indeed be turned into a rigorous mathematical
statement in one-period models. More specifically, they show
that lower risk aversion leads to a payoff which is larger in terms of second
order stochastic dominance.

In the present study, we extend their results to (weakly) complete
continuous-time models. We also complement an ad-hoc counterexample of
Dybvig and Wang, by showing that these results are ``fragile'', in the
sense that they fail in essentially \emph{any} model, if the latter is
perturbed on a set of arbitrarily small probability.
 On the other hand, we establish that they hold for power investors in models with (conditionally) independent
increments. \\

 \end{abstract}
\maketitle

\section{Introduction}
A classical problem in mathematical finance and financial economics is to maximize expected utility from terminal wealth. This means that -- given a time horizon $T$ and a utility function $U$ describing the investor's preferences -- one tries to choose a trading strategy such that the terminal value $\hat{X}_T$ of the corresponding wealth process maximizes $\expval{U(X_T)}$ over all wealth processes of competing strategies. Existence and uniqueness of the maximizer $\hat{X}$ are assured in very general models and under essentially minimal assumptions (cf., e.g., \cite{kramkov.schachermayer.99} and the references therein). However, much less is known about the qualitative properties of $\hat X$ and, in particular, about their dependence on the investor's attitude towards risk measured, e.g., in terms of her \emph{absolute risk aversion} $-U''/U'$.

Since comparative statics for the composition of the investor's portfolio are impossible to obtain in any generality, Dybvig and Wang (\cite{dybvig_wang09}, henceforth DW) have recently proposed to compare the \emph{distributions} of the optimal payoffs instead. They show that -- in one-period models -- the payoffs of investors with ordered absolute risk aversion can be ranked in terms of stochastic dominance relationships. More specifically, suppose investor $L$ is less risk averse than the more risk inverse investor $M$,  and the corresponding optimal payoffs $\hat X^M_T, \hat X^L_T$ have finite first moments. Then \cite[Theorems 3 and 7]{dybvig_wang09} assert that $\hat X^L_T$ dominates $\hat X^M_T$ in the \emph{monotone convex order},
\begin{align}\label{WeakResult} \hat X^M_T\leq_{MC} \hat X^L_T,\end{align} 
that is, $\E [c(\hat X^M_T)]\leq \E [c(\hat X^L_T)]$ for every monotone increasing convex function $c:\R_+\to \R$. Moreover, if either of the utility functions has nonincreasing absolute risk aversion, then DW also obtain the sharper assertion that
  \begin{align}\label{EResult}
    \E [\hat X^M_T] & \leq \E [\hat X^L_T] \quad\mbox{and}\\
    (\hat X^M_T - \E [\hat X^M_T]) & \leq_{C}(\hat X^L_T- \E [\hat X^L_T]), \label{StrongResult}
  \end{align}
where $\leq_C$ denotes the \emph{convex order}, i.e., \eqref{StrongResult} asserts that, for every convex function $c$, $\E[c(\hat X^M_T-\E [\hat X^M_T])]\leq \E [c(\hat X^L_T-\E [\hat X^L_T])]$. Both $\leq_{C}$ and $\convmon$ are \emph{second order stochastic dominance} relations.

By Strassen's characterization of the convex order, cf.\ \cite{strassen65existence}, \eqref{StrongResult} is tantamount to the existence of a random variable $\eps$ with $\E[\eps|\hat X^M_T]=0$ such that, in distribution, 
\begin{align}\label{StrassenStrong}
\hat X^L_T = \hat X^M_T+ (\E[\hat X^L_T]-\E[\hat X^M_T]) +\eps.
\end{align}
In plain English, this means that the less risk averse investor is willing to accept the extra noise $\eps$ in exchange for the additional risk premium $\E[\hat X^L_T]-\E [\hat X^M_T] \geq 0$. 

In addition, DW also construct some counterexamples showing that the above results generally do not hold in incomplete markets.

The purpose of the present study is threefold. Firstly, we prove an analogue of the main result of DW -- which is stated in the discrete one-period setting common in much of economics -- in the continuous-time framework prevalent in mathematical finance, under the assumption that the market is (weakly) complete. Whereas it would also be possible to extend the approach of DW, we believe that our presentation is both more compact and more transparent.

Next, in Section 3, we shed more light on the fragility of this structural result in incomplete markets. Whether the counterexamples of DW use somewhat ad-hoc models and utility functions, we show that -- even for investors with power utilities -- the result does not hold in \emph{any} finite state model, if the latter is perturbed by adding just a single extra branch with arbitrarily small probability. 

Finally, in Section 4, we take a look at additional structural assumptions which ensure the validity of DW's result also in incomplete markets. More specifically, we show that it holds for power utility investors if the increments of the return processes are independent or, more generally, independent conditional on some stochastic factor process. We emphasize that even though the utility maximization problem can be solved fairly explicitly in this setup, the stochastic dominance relationship apparently cannot be read off the formulas. Instead, we prove the result by induction in a discrete approximation of the model and then pass to the limit. An extension of this result to more general preferences and/or market models appears to be a challenging topic for future research.

\section{(Weakly) Complete Markets}

Fix a filtered probability space $(\Omega,\mathcal{F}, (\mathcal{F}_t)_{t \in [0,T]},\mathbb{P})$. We consider a market of one riskless and $d$ risky assets and work in discounted terms. That is, the riskless asset is supposed to be normalized to $1$, whereas the (discounted) price process of the risky asset is assumed to be modeled by an $\mathbb{R}^d$-valued semimartingale $S$.

\subsection{Utilities defined on the positive halfline}

 The investor's preferences are described by a \emph{utility function}. Here we first consider the case where the latter is defined on the positive halfline. That is, it is assumed to be a strictly increasing, strictly concave, twice differentiable mapping $U: \mathbb{R}_+ \to \mathbb{R} \cup \{\infty\}$ satisfying the Inada conditions $\lim_{x\to\infty}U'(x)=0$ and $\lim_{x\to 0}U'(x)=\infty$.  Given a utility function $U$, the quotient $-U''(x)/U'(x)$ is called the \emph{absolute risk aversion} of $U$ at $x\in(0,\infty)$, cf.\ \cite{pratt64risk,arrow65aspects}. An investor with utility function $U_\agent{RA}$ is called \emph{more risk averse} than an agent with utility function $U_\agent{RL}$, written $U_\agent{RA}\mra U_\agent{RL}$, if the absolute risk aversion of $U_\agent{RA}$ dominates the absolute risk aversion of $U_\agent{RL}$ pointwise. In the sequel we frequently use that $U_\agent{RA}\mra U_\agent{RL}$ if and only if $U'_\agent{RL}(x)/U'_\agent{RA}(x)$ is monotone increasing for all $x\in(0,\infty)$.

For the remainder of this section, we suppose that the market is \emph{complete}, i.e., that the set of equivalent (local) martingale measures is a singleton $\mathbb{Q}$,\footnote{In fact, an inspection of the proofs shows that it is sufficient to assume that the dual minimizer of \cite{kramkov.schachermayer.99} is the same for both agents. If this holds for \emph{all} agents, this property has been called \emph{weak completeness} of the financial market, see \cite{kramkov_sirbu06sensitivity, schachermayer_sirbu09financial}.} and consider the problem of maximizing expected utility, $\sup_{X_T}\expval{U(X_T)}$. Here, $X_T$ runs though the terminal values of all wealth processes $X$ that can be generated by self-financing trading starting from an initial endowment $x>0$, and satisfy the admissibility condition $X \geq 0$. Throughout, we suppose that the supremum is finite, as, e.g., for utility functions that are bounded from above. Then, it is well-known (cf., e.g., \cite[Theorem 2.0]{kramkov.schachermayer.99}) that there is a unique optimal wealth process $\hat{X}$ related to the martingale measure $\mathbb{Q}$ via the first-order condition
\begin{equation}\label{eq:duality}
U'(\hat{X}_T)=y\frac{d\mathbb{Q}}{d\mathbb{P}}.
\end{equation}
Here, the Lagrange multiplier $y$ is a constant given by the marginal indirect utility of the initial capital $x$ (cf., e.g., \cite{kramkov.schachermayer.99} for more details).

For a more risk averse investor with utility function $U_M$ and a less risk averse investor with utility function $U_L$, we are now able to state our first main result, the counterpart of \cite[Theorem 3]{dybvig_wang09} in continuous time: Lower risk aversion leads to a terminal payoff that is larger in the monotone convex order \eqref{WeakResult}. 

\begin{theorem}\label{thm:monotone_convex_order}
  Let $\hat X^\agent{RA}_T$, $\hat X^\agent{RL}_T$ be integrable and suppose that $U_\agent{RA}\mra U_\agent{RL}$. Then 
  $$\hat X^\agent{RA}_T\convmon \hat X^\agent{RL}_T.$$
\end{theorem}

If the absolute risk aversion of at least one agent is \emph{nonincreasing}\footnote{E.g., this holds for investors with power utility functions $x^{1-p}/(1-p)$, i.e, with constant relative risk aversion $0 < p \neq 1$.} we also obtain the stronger convex order result \eqref{StrongResult}. 

\begin{theorem}\label{thm:convex_order}
  Let $\hat X^\agent{RA}_T$, $\hat X^\agent{RL}_T$ be integrable and suppose that $U_\agent{RA}\mra U_\agent{RL}$. If, in addition, either $U_\agent{RA}$ or $U_\agent{RL}$ has nonincreasing absolute risk aversion, then 
  $$(\hat X^M_T - \E[\hat X^M_T]) \leq_{C}(\hat X^L_T- \E[ \hat X^L_T]).$$
\end{theorem}

To simplify what has to be proved we use the following well-known characterization of the (monotone) convex order, which is a straightforward consequence of the monotone convergence theorem.

\begin{lemma}\label{lem:convmon_mon}
Let $X,Y$ be random variables with finite first moments. 
\begin{enumerate}
\item We have $X\leq_{MC} Y$ if and only if $\expval{(X-K)^+}\leq \expval{(Y-K)^+}$ for all $K\in \R$.
\item We have  $X\leq_C Y$ if and only if $X\leq_{MC} Y$ and $\expval{X}=\expval{Y}$.
\end{enumerate}
\end{lemma}

\begin{proof}[Proof of Theorem \ref{thm:monotone_convex_order}]
To simplify notation, first notice that we may assume $y_\agent{RA}=y_\agent{RL}=1$. Indeed this is achieved by rescaling $U_M$ and $U_L$ by the factor $y_M$ and $y_L$, respectively, which has no affect on the utility maximization problem and the risk aversion of the utility functions. Setting $D:=d\mathbb{Q}/d\mathbb{P}$, $F:= (U_M')^{-1}$, and  $G:= (U_L')^{-1}$,  the first-order condition \eqref{eq:duality} can be rewritten as 
\begin{align}\label{DRep} \hat X_T^M= F(D), \quad \hat X_T^L= G(D) \end{align}
and $U_\agent{RA}\mra U_\agent{RL}$ implies that $ {F(x)}/{G(x)}$ is decreasing in $x$. 

Next, notice that there exists $q\in\R$ such that, almost surely,
\begin{align}\label{NicePoint}
q\leq \hat X^M\leq \hat X^L \quad \mbox{or}\quad 
q\geq \hat X^M\geq \hat X^L.
\end{align}
To see this consider $\rho:=D(\Q)$. Since the value processes are $\Q$-martingales by \cite[Theorem 2.0]{kramkov.schachermayer.99} and have the same initial value $x$, it follows that
$$\textstyle{\int F\, d\rho=\int F(D)\, d\Q= \int \hat X^M\, d\Q=\int \hat X^L\, d\Q=\int G(D)\, d\Q= \int G\, d\rho}.$$ 
As $F$ and $G$ are continuous this implies that there exists $p>0$ such that $F(p)= G(p)=:q$. Since $F/G$ is decreasing, we obtain for all $x\in \R$ that either $q \leq F(x)\leq G(x)$ or  $q\geq F(x)\geq G(x),$ which yields \eqref{NicePoint}.

As $U_M'$ and $U_L'$ are decreasing, we deduce from \eqref{NicePoint} and $U_M'(\hat X^M)=D=U_L'(\hat X^L)$ that $(D-p)(\hat X^\agent{RL}-\hat X^\agent{RA})\leq 0$. The identity
  \begin{align*}
    0=\mathbb{E}_\mathbb{Q}\big[\hat X^\agent{RL}-\hat X^\agent{RA}\big]={p}\E\big[\hat X^\agent{RL}-\hat X^\agent{RA}\big]+\E\big[(D-{p})(\hat X^\agent{RL}-\hat X^\agent{RA})\big]
  \end{align*}
now yields the intermediate  result $\E[\hat X^M]\leq\E[\hat X^L]$.

It remains to establish that $\E[(\hat X^M-K)^+]\leq \E[(\hat X^L-K)^+]$ for $K\in \R$. If $K\geq q$ this is a trivial consequence of \eqref{NicePoint}.
 
If $K\leq q$, then \eqref{NicePoint} implies $\E[(\hat X^M-K)^-]<\E[(\hat X^L-K)^-]$. Adding the inequality $\E[\hat X^M-K]\leq\E [\hat X^L-K]$ we obtain the desired relation $\E[(\hat X^M-K)^+]\leq \E[(\hat X^L-K)^+]$  also in this case.
\end{proof}

The proof of Theorem \ref{thm:convex_order} follows a similar scheme.

\begin{proof}[Proof of Theorem \ref{thm:convex_order}]
By Lemma \ref{lem:convmon_mon} it suffices to show $\hat X^\agent{RA}\leq_{MC} \hat X^\agent{RL}-l$ where $l:=\E[\hat X^\agent{RL}-\hat X^\agent{RA}]$. 
Note that $l>0$ by the proof of Theorem \ref{thm:monotone_convex_order}. 

Since either $U_\agent{RA}$ or $U_\agent{RL}$ has non increasing risk aversion, $U_\agent{RL}'(x+l)/U_\agent{RA}'(x)$ is increasing in $x$. As above we may assume $y_\agent{RA}=y_\agent{RL}=1$. Setting $F:=(U_\agent{RA}')^{-1}, \tilde G:=(U_\agent{RL})^{-1}+l$ and  $ \tilde\rho:=D(\P)$ we have $\int F\, d\tilde \rho =\int \tilde G\, d\tilde \rho$. Arguing as before, we obtain the existence of a point $\tilde q$ such that, a.s., 
$$\tilde q\leq \hat X^M\leq \hat X^L-l \quad \mbox{or}\quad \tilde q\geq\hat  X^M\geq\hat  X^L-l$$
in analogy to \eqref{NicePoint}.
 
As in the last step of the above proof of Theorem \ref{thm:monotone_convex_order}, this implies that $\E[(\hat X^M-K)^+]\leq \E[((\hat X^L-l)-K)^+]$ for all $K\in\R$.
\end{proof}
\begin{remark}  
%
  The converse of Theorem \ref{thm:monotone_convex_order} also holds true: If two agents choose -- in \emph{every} complete market --  payoffs $\hat X_T^\agent{RA}$, $\hat X_T^\agent{RL}$ satisfying
  \begin{equation}\label{eq:converse}
    \hat X_T^\agent{RA}\convmon \hat X_T^\agent{RL}
  \end{equation}
 then their corresponding utility functions satisfy $U_\agent{RA}\mra U_\agent{RL}$. This is a direct consequence of \cite[Theorem 4]{dybvig_wang09}, which establishes the above statement under the weaker assumption that \eqref{eq:converse} holds for all complete \emph{one-period} market models.
\end{remark}

\subsection{Utilities defined on the entire real line}

We now turn to investors with utility functions defined on the whole real line. Whereas the final results are analogous, the necessary definitions are technically more involved.

In this setting, we assume that the asset price process $S$ is locally bounded.  A \emph{utility function} then is a strictly increasing, strictly concave, twice differentiable mapping $U: \mathbb{R} \to \mathbb{R} \cup \{\infty\}$ satisfying both the Inada conditions $\lim_{x\to\infty}U'(x)=0$ and $\lim_{x\to -\infty}U'(x)=\infty$ and the condition of reasonable asymptotic elasticity:
$$\limsup_{x \to \infty} \frac{xU'(x)}{U(x)}<1 \quad \mbox{and} \quad \liminf_{x \to -\infty} \frac{xU'(x)}{U(x)}>1.$$

Following \cite{schachermayer.02}, the wealth process $X$ of a self-financing trading strategy starting from an initial endowment $x \in \mathbb{R}$ is called \emph{admissible}, if its utility $U(X_T)$ is integrable, and it is a supermartingale under all absolutely continuous local martingale measures $\mathbb Q$ with ``finite $V$-expectation'', $\expval{V(d\mathbb Q/d\mathbb P)}<\infty$, for the conjugate function $V(y)=\sup_{x \in \mathbb{R}}(U(x)-xy)$, $y>0$, of $U$. Throughout, we suppose that the market admits an equivalent local martingale measure (i.e., satisfies NFLVR) and that for each $y>0$, the dual problem $\inf_{\mathbb{Q}} \expval{V(yd\mathbb Q/d\mathbb P)}$ is finite with a dual minimizer $\hat{\mathbb Q}(y)$ in the set of equivalent local martingale measures. Sufficient conditions for the validity of the latter assumption can be found in \cite{bellini.fritelli.02}; in particular it holds if the market is complete or if the utility function under consideration is exponential, $U(x)=-e^{-\gamma x}$ with $\gamma>0$, and an equivalent local martingale measure $\mathbb Q$ with finite entropy $\expval{d\mathbb{Q}/d\mathbb P\log(d\mathbb Q/d\mathbb P)}<\infty$ exists.

Subject to these assumptions, \cite[Theorem 1]{schachermayer.02} ensures that there is a unique wealth process $\hat X$ that maximizes utility from terminal wealth. Moreover, for a suitable Lagrange multiplier $y$, the latter is once again related to the corresponding dual minimizer $\hat{\mathbb{Q}}(y)$ via the first-order condition 
$$y\frac{d\hat{\mathbb{Q}}(y)}{d\mathbb{P}}=U'(\hat X_T).$$
This was the key property for our proof of Theorems \ref{thm:monotone_convex_order} and \ref{thm:convex_order}. Indeed, an inspection of the proofs shows that we never used that the domains of the utility functions are given by the positive halfline. Hence, we obtain the following analogous results:

\begin{theorem}\label{thm:monotone_convex_order_exp}
Consider two agents with utility functions  $U_\agent{RA}, U_\agent{RL}$ defined on the whole real line and suppose the corresponding optimal terminal payoffs  $\hat X^\agent{RA}_T$, $\hat X^\agent{RL}_T$ are integrable. Then if $U_\agent{RA}\mra U_\agent{RL}$ and the dual minimizers for both agents coincide, we have 
$$\hat X^\agent{RA}_T\convmon \hat X^\agent{RL}_T.$$
If, in addition, either $U_\agent{RA}$ or $U_\agent{RL}$ has nonincreasing absolute risk aversion, then 
$$(\hat X^M_T - \E[\hat X^M_T]) \leq_{C}(\hat X^L_T- \E[ \hat X^L_T]).$$  
\end{theorem}

Since the so-called \emph{minimal entropy martingale measure} is the dual minimizer for all exponential utility maximizers -- irrespective of risk aversion and initial endowment -- it follows that the above result is always applicable in this case. That is, no extra assumptions other than the integrability of the agents optimal payoffs need to be imposed on the financial market.

\begin{corollary}
Consider two agents with exponential utilities $-e^{-\gamma_\agent{RA}x}$ resp.\ $-e^{-\gamma_\agent{RL}x}$. Then if $\gamma_\agent{RA} > \gamma_\agent{RL}$, the assumptions of both parts of Theorem \ref{thm:monotone_convex_order_exp} are always satisfied, provided that the agents optimal payoffs are integrable.
\end{corollary}

\begin{proof}
First notice that exponential utilities have constant and therefore nonincreasing absolute risk aversion. Next note that the notion of admissibility is both independent of the initial endowment and scale invariant for $U(x)=-e^{-\gamma x}$. Hence the optimal strategy is evidently independent of the initial endowment and inversely proportional to the absolute risk aversion $\gamma$. Since the Lagrange multiplier $y$ is given by the marginal indirect utility (cf., e.g., \cite[Theorem 1]{schachermayer.02}), it then follows from the first-order condition that the dual minimizer is the same for all absolute risk aversions $\gamma$.
\end{proof}

\section{Structural Counterexample in Incomplete Markets}

In this section we show that -- even in finite probability spaces and for investors with power utility functions --  Theorems \ref{thm:monotone_convex_order} and \ref{thm:convex_order} are ``fragile'', in that they do not hold in \emph{any} market model if the latter is perturbed by adding a single extra state with arbitrary small probability. 
   
  \subsection{Basic Idea} 
  
Our starting point is the simple observation that the monotone convex order between random variables $X$, $Y$ can be destroyed by minimal perturbations of the distributions of $X$, $Y$, see Figure \ref{fig:convex_order_counterexample}.
  
To enforce that $\hat X_T^M\not\leq_{MC}\hat X_T^L$ it is
sufficient\footnote{Indeed, if $\hat X_T^M\not\leq_{MC}\hat X_T^L$ this is
always witnessed by a hockeystick function, cf.\ Lemma
\ref{lem:convmon_mon}.} to find some number $K^*\in\R$ such that
$\E[(\hat X_T^M-K^*)^+]>\E[(\hat X_T^M-K^*)^+]$.
That is, in order to construct a model for which the montone convex order
\eqref{WeakResult} fails, we want to assure that the more risk averse
agent \agent{RA} attains with a higher probability than \agent{RL}
large values above a certain threshold $K^*$.  In particular, in
finite probability spaces $\max_\omega
\hat X_T^\agent{RA}(\omega)>\max_\omega \hat X_T^\agent{RL}(\omega)$ already
assures that the montone convex order fails; this follows by
considering the call with $K^*= \max_\omega \hat X_T^\agent{RL}$.
 
  \begin{figure}[t]
    \centering
    \subfigure[montone convex order]{
      \includegraphics[width=0.45\textwidth]{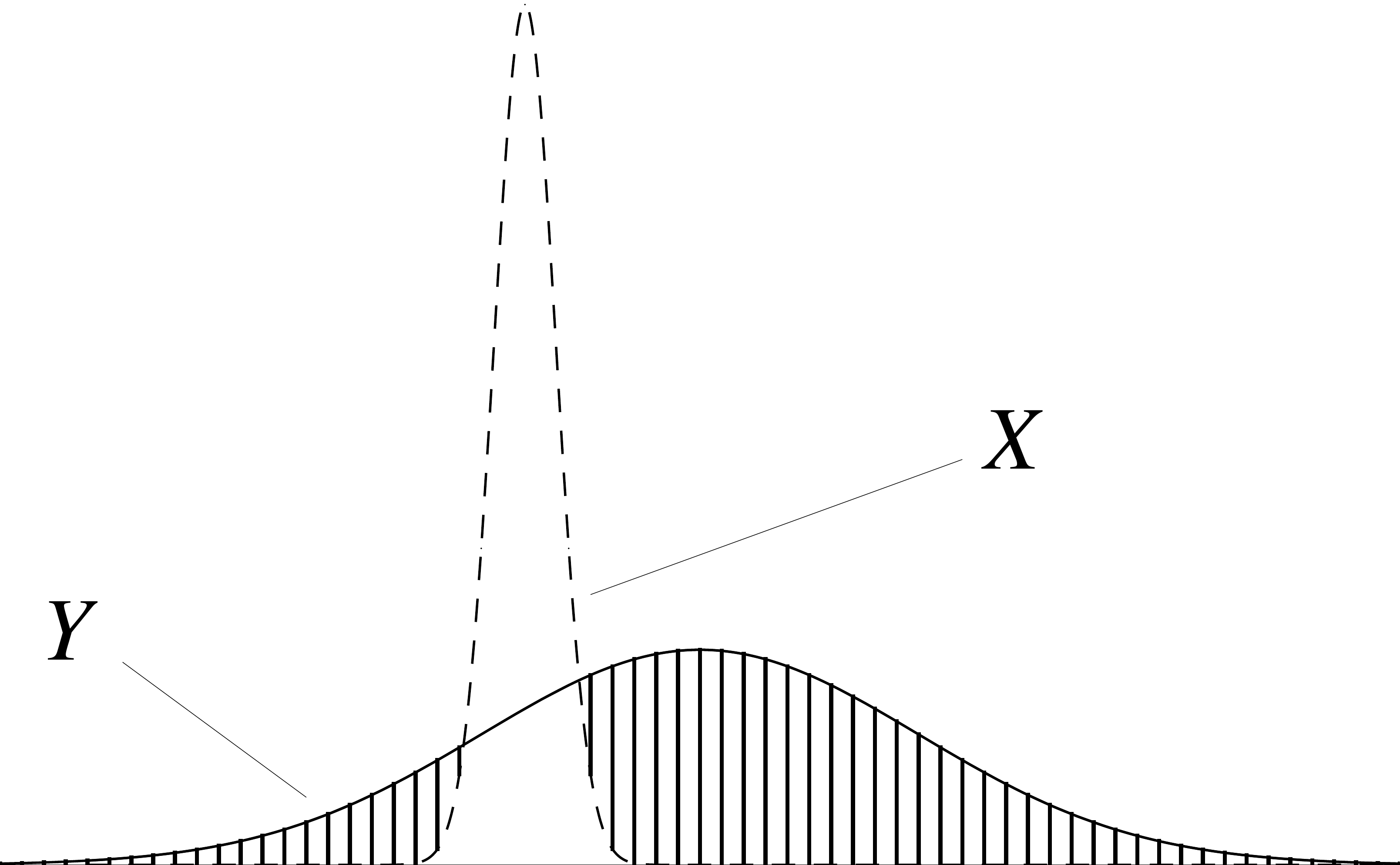}
      \label{fig:convex_order}
    }
    \subfigure[no monotone convex order]{
      \includegraphics[width=0.45\textwidth]{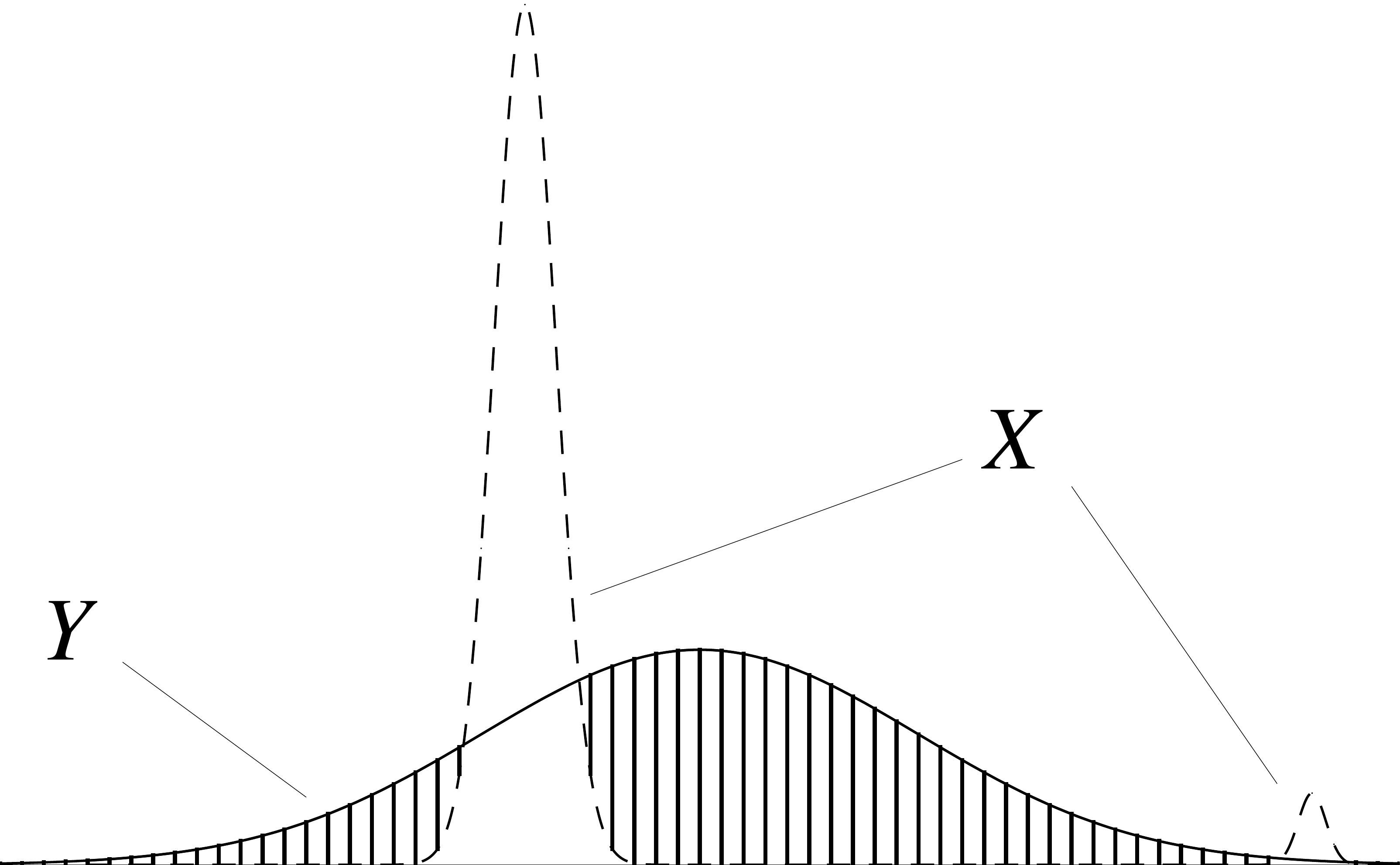}
      \label{fig:counterexample}
    }
      \caption{Both illustriations show the distributions of random variables $X$ and $Y$. \subref{fig:convex_order}: $X\convmon Y$; \subref{fig:counterexample}: $X\nleq_{MC}Y$}                
    \label{fig:convex_order_counterexample}
  \end{figure}

  \subsection{Concrete Counterexample}
  We now present an explicit example for one riskless and one risky asset, showing that lower risk aversion does in general not lead to a larger portfolio in the monotone convex order. Moreover, our example exemplifies that this can happen even if the less risk averse investor always invest a larger fraction of her wealth in the risky asset. The latter is the decisive property used by DW for the proof of their results in incomplete one-period models.
 
  We start with a complete two-period market model for which the monotone convex order \eqref{WeakResult} holds true by Theorem \ref{thm:monotone_convex_order} above. We then alter this model by inserting a new branch after the first period which makes the model incomplete. The new branch occurs with an arbitrarily small probability $\epsilon$ so that  the optimal strategies in the new model are almost identical to the original strategies. However, this new branch is constructed in such a way that the more risk averse agent \agent{RA} attains with positive probability a payoff that is larger than any possible payoff of \agent{RL}, which implies $\hat X_T^\agent{RA}\nleq_{CM}\hat X_T^\agent{RL}$.

  For simplicity we consider a complete binomial  model  for which the stock does not change in the second period but stays constant. 
  
  \begin{center}
    \begin{tikzpicture}[>=stealth,sloped]
      \matrix (tree) [%
	matrix of nodes,
	minimum size=1cm,
	column sep=1.5cm,
	row sep=0.0cm,
      ]
      {
		& $S_1=2$ & $S_2=2$ \\
	$S_0=1$	&  &  \\
		&  &   \\
		&  &  \\
		& $S_1=0.5$  & $S_2=0.5$\\ 	
      };
      \draw[->] (tree-2-1) -- (tree-1-2) node [midway,above] {$0.6$};
      \draw[->] (tree-2-1) -- (tree-5-2) node [midway,below] {$0.4$};
      \draw[->] (tree-1-2) -- (tree-1-3) node [midway,above] {$1$};     
      \draw[->] (tree-5-2) -- (tree-5-3) node [midway,below] {$1$};
    \end{tikzpicture}
  \end{center}
  The preferences of agents \agent{RA} and \agent{RL} are given by power utilities with relative risk aversions $p_\agent{RA}=0.9$ and $p_\agent{RL}=0.3$, respectively.\footnote{Indeed any other choice of $p_i>0$ can be made to work as well.} By direct computation, we find that the agents optimally invest fractions $\hat{\pi}^\agent{RA}_0\approx0.887$ resp.\ $\hat{\pi}^\agent{RL}_0\approx 1.853$ of their wealth in the risky asset at time $t=0$. In particular, since $\agent{RL}$ invests a larger fraction of wealth than \agent{RA}, agent \agent{RL} has less money than \agent{RA} when the price of the risky asset  decreases.
  
   The branch that we want to insert into $S$ after the first period takes advantage of this disparity of wealth between \agent{RA} and \agent{RL}. It is given by $S^*$, 
 \begin{center}  
  \begin{tikzpicture}[>=stealth,sloped]
      \matrix (tree) [%
	matrix of nodes,
	minimum size=1cm,
	column sep=1.5cm,
	row sep=0.0cm,
      ]
      {
			  &  \\
			  & $S_1^*=0.5K$ \\
	$S_0^*=0.5$ &   \\
		& $S_1^*=0.25$ \\
		&   \\		
      };
      \draw[->] (tree-3-1) -- (tree-2-2) node [midway,above] {$1-\alpha$};
      \draw[->] (tree-3-1) -- (tree-4-2) node [midway,below] {$\alpha$};
  \end{tikzpicture}
 \end{center}
 where we fix $\alpha$ to be a small probability, say  $\alpha= 0.05$ and $K$ a large stock value, for instance $K=20$.
 Since this market offers both agents a very high probability of big fortune, \agent{RA} and \agent{RL} invest almost as much as admissibility allows and choose $\hat\pi^{\agent{RA}*}\approx1.9492$, respectively $\hat\pi^{\agent{RL}*}\approx 1.9999$ as their optimal fraction of wealth invested in $S^*$.
  
  We now ``perturbate'' $S$ and define a new process $S'$ by including $S^*$ into $S$ after the first step, i.e., we fix a small probability  $\epsilon=0.01$ and define $S'$ by
  \begin{center}  
    \begin{tikzpicture}[>=stealth,sloped]
	\matrix (tree) [%
	  matrix of nodes,
	  minimum size=1cm,
	  column sep=2cm,
	  row sep=0.0cm,
	]
	{
		  & $S'_1=2$ & $S'_2=2$ \\
	  $S'_0=1$	&  &  $S'_2=0.5K$\\
		  & $S'_1=0.5$ &   \\
		  &  &  $S'_2=0.25$\\
		  & $S'_1=0.5$  & $S'_2=0.5$\\ 	
	};
	\draw[->] (tree-2-1) -- (tree-1-2) node [midway,above] {$\approx0.6$};
	\draw[->] (tree-2-1) -- (tree-5-2) node [midway,below] {$\approx0.4$};
	\draw[->] (tree-2-1) -- (tree-3-2) node [midway,below] {$\epsilon$};
	\draw[->] (tree-1-2) -- (tree-1-3) node [midway,above] {$1$};     
	\draw[->] (tree-5-2) -- (tree-5-3) node [midway,below] {$1$};
	\draw[->] (tree-3-2) -- (tree-2-3) node [midway,below] {$1-\alpha$};
	\draw[->] (tree-3-2) -- (tree-4-3) node [midway,below] {$\alpha$};
    \end{tikzpicture}
  \end{center}
  By the dynamic programing principle, the optimal trading strategies $\hat\pi_1'$ chosen at $t=1$ for the new model $S'$ are given by $\hat\pi^{\agent{RA}*}$, resp.\ $\hat \pi^{\agent{RL}*}$ above.\footnote{In a general two period model the trading strategy chosen at $t=1$ would depend on the current state of the first period. Since $S'_2$ only changes in the branch given by $S^*$ we can neglect this dependence.} Since $\epsilon$ is small, the optimal trading strategies $\hat\pi'_0$ chosen at $t=0$ are close to $\hat\pi^\agent{RA}_0$, resp.\ $\hat\pi^\agent{RL}_0$ above, and can be numerically computed to be given by $\hat\pi^{\agent{RA}'}_1\approx0.8595$, resp.\ $\hat\pi^{\agent{RL}'}_1\approx1.6622$.

  We thus see that \agent{RL} invests in the second step a larger fraction of wealth in the stock than \agent{RA}. But since \agent{RA}'s wealth after the first period is larger when $S_1'=0.5$, \agent{RA} invests more \emph{money} in $S'$ than \agent{RL}. In particular we find that the optimal terminal payoffs $\hat X_2^M, \hat X_2^L$ satisfy $21.6897 \approx \max \hat X^\agent{RA}_2(\omega)>\max_\omega \hat X^\agent{RL}(\omega) \approx 6.5873$, where the maxima are attained at the event $S'_2=0.5K$. Hence, $\hat X^\agent{RA}_2\nleq_{MC}\hat  X^\agent{RL}_2$.

\subsection{Counterexample in an $n$-period model}
In the previous section we have seen in a concrete example that small changes of the model can cause the failure ofthe convex order relationship  \eqref{WeakResult}. Indeed this applies in a much wider setting; here we want to illustrate this in the case of an (arbitrage free) $n$-period  model $(S_i)_{i=0}^n$ (where $n\geq 2$), defined on a finite probability space.     

Consider, once again,  agents $M,L$  equipped with power utility functions with parameters $p_\agent{RL}=0.3$ resp.\ $p_\agent{RA}=0.9$. Assume for simplicity that the stock price does not stay constant during any period. In this case the investors face a strictly concave optimization problem, hence, the optimal strategies $(\hat\pi_i^M)_{i=1}^n$, resp.\ $(\hat\pi_i^L)_{i=1}^n$ are uniquely determined. Denote by $\hat X^M$ resp.\ $\hat X^L$ the resulting optimal wealth processes. 

We make the further assumption that the respective optimal wealth processes are not equal, more precisely that $\hat X^M_{n-1}$ is not equal to  $\hat X^L_{n-1}$. 

Fix an arbitrary small number $\eta>0$. 
Then it is possible to replace the process $S$ by a new process $ S'$ which  agrees with $S$ during the first $n-1$ stages and differs from $S$ only in the last stage and with probability less than $\eta$, but for which \eqref{WeakResult} fails. 

Using the assumption that  $\hat X^M_{n-1}\neq \hat X^L_{n-1}$ we find that there exist
$a,b,c\in \R$, $a>b$ such that the  event  $A=\{\hat X^M_{n-1}=a, \hat X^L_{n-1}=b, S_{n-1}=c\}$ has positive probability.  

We now introduce a coin flip $\theta$, independent of the stock price model and so that the  outcome is $\{\theta=\mbox{head}\}$ with probabilty $\eps<\eta$ and $\{\theta=\mbox{tail}\}$ with probability $1-\epsilon$.
If the coin shows tail then the  stock price process remains unchanged, i.e.\ $S'=S$. But in the event $A\cap\{\theta =\mbox{head}\}$, the stock price process in the last period is, as above, replaced by $S^*$ given through 
 \begin{center}  
  \begin{tikzpicture}[>=stealth,sloped]
      \matrix (tree) [%
	matrix of nodes,
	minimum size=1cm,
	column sep=1.5cm,
	row sep=0.0cm,
      ]
      {
			  &  \\
			  & $S_{n}^*=cK$ \\
	$S_{n-1}^*=c$ &   \\
		& $S_n^*= \frac{c}{2}$ \\
		&   \\		
      };
      \draw[->] (tree-3-1) -- (tree-2-2) node [midway,above] {$1-\alpha$};
      \draw[->] (tree-3-1) -- (tree-4-2) node [midway,below] {$\alpha$};
  \end{tikzpicture}
 \end{center}
where $K>1$. An elementary analysis of the above example reveals that for $\alpha$ sufficiently close to $1$ both agents will invest almost as much in the stock as admissibility allows. As $a>b$ we can arrange the constant $K$ large enough so that the maximal payoff of agent $M$ supersedes that of agent $L$ as well as the maximum of $\hat X_n^L$.

There remains one issue to be dealt with: due to the change in the model we are now facing new optimal strategies and value processes. To cope with this problem we observe that the orginal problem was only changed on the portion $A\cap\{\theta =\mbox{head}\}$ of our space which has probability at most $\eps$ and that the possible gains on this set are bounded by some constant independent of $\eps$. Consequently the perturbation of the original model vanishes as $\eps\to 0$. As the original maximization problems had unique solutions, the new optimal strategies resemble the original ones as closely as we want (with the notable exception of the event $A\cap\{\theta =\mbox{head}\}$, in period $n-1$). 

Summing up, upon choosing $\eps>0$ sufficiently small we obtain for the new optimal terminal wealth $\max_\omega \hat X_n^\agent{RA}(\omega)>\max_\omega\hat X_n^\agent{RL}(\omega)$ and in particular that \eqref{WeakResult} fails. 

We conclude this section by pointing out that our argument is still valid if the stock is allowed to stay constant. Indeed, in this case the uniqueness of the optimal trading strategies is only violated in periods for which the stock does not change, but this does not affect the above reasoning. 

\section{Models with (Conditionally) Independent Returns}
 
 In this last section, we consider some particular incomplete markets in continuous time, where the results of DW do hold. 
 
 More specifically, we focus on power utility investors in models of one riskless and one risky asset with independent or, more generally, conditionally independent returns. Its price process is first assumed to be modeled as the stochastic exponential $S=\mathcal{E}(R)$ of a \emph{L\'evy process} $R$, i.e, $dS_t/S_t=dR_t$. By its very definition, the L\'evy process $R$ has independent (and in fact also stationary) increments. Here, it can be interpreted as the returns process that generates the price process $S$ of the risky asset in a multiplicative way.
 
Concerning preferences, we focus on investors with power utilities, i.e., $U(x)=x^{1-p}/(1-p)$, where $0<p\neq 1$ denotes the investor's constant relative risk aversion. In this case, trading strategies are most conveniently parametrized in terms of the \emph{fractions} $\pi_t$ of wealth invested in the risky asset at time $t \in [0,T]$ (cf., e.g., \cite{nutz09power} for a careful exposition of this matter). The wealth process corresponding to the risky fraction process $(\pi_t)_{t \in [0,T]}$ is then given by $dX_t/X_t=\pi_tdR_t$, i.e., $X_t=x\mathcal{E}(\int_0^\cdot \pi_s dR_s)_t$.

In this setting, it has been proved -- by \cite{samuelson69lifetime} in discrete time and, in increasing degree of generality, by \cite{framstad.al.99, benth_karlsen01optimal, kallsen00optimal,nutz09power} in continuous time -- that the optimal policy is to invest a \emph{constant} fraction $\hat{\pi}$ in the risky asset. The latter is known implicitly as the maximizer of some deterministic function, see \cite{nutz09power}. In addition, it is possible to obtain some comparative statics for the optimal risky fractions here. More specifically, for two power utility functions $U_\agent{RA}\mra U_\agent{RL}$,\footnote{I.e., the relative risk aversion $p^M$ of $M$ is larger than its counterpart $p^L$ for $L$.} the optimal risky fractions $\hat\pi_\agent{RA}$, $\hat\pi_\agent{RL}$ satisfy $|\hat\pi_\agent{RA}|\leq|\hat\pi_\agent{RL}|$ and are non-negative (non-positive) if $\expval{R_t}$ is non-negative (non-positive) for some (or equivalently all) $t$, cf.\ \cite[Proposition 4.4]{temme11notes}. 
 
The main result of this section is stated in the following theorem. 

   \begin{theorem} \label{thm:conv_stoch_dom_levy}
   Suppose $R$ is square-integrable and neither a.s.\ increasing nor a.s.\ decreasing, and $S=\mathcal{E}(R)$ is strictly positive. Then for power utility functions $U_\agent{RA}\mra U_\agent{RL}$ the optimal payoffs $\hat X^\agent{RA}$, $\hat X^\agent{RL}$satisfy $$\hat X^\agent{RA}_T-\E[\hat X^\agent{RA}_T]\conv \hat X^\agent{RL}_T-\E[\hat X^\agent{RA}_T].$$
  \end{theorem}
  
  Since the optimal fractions are at least known implicitly as the maximizers of a scalar function and, in particular, satisfy $|\hat\pi_\agent{RA}|\leq|\hat\pi_\agent{RL}|$, one might think that this result can be obtained by a direct comparison of the corresponding wealth processes $\hat X^\agent{RA}_T= \mathcal{E}(\hat{\pi}_\agent{RA}R)_T$ and $\hat X^\agent{RL}_T=\mathcal{E}(\hat{\pi}_\agent{RL}R)_T$. However, the dependence of these random variables on the risky fractions is quite involved as can be seen by looking at the explicit formula for the stochastic exponential \cite[Theorem I.4.61]{jacod_shiryaev03limit}.   
  
 In a discrete-time setting, Theorem \ref{thm:conv_stoch_dom_levy} can be established by induction using the results of Kihlstrom, Romer, and Williams \cite{kihlstrom.al.81} and the scaling properties of the power utilities. However, the corresponding result in continuous time cannot generally be obtained by passing to the limit since the continous-time optimizer can lead to bancruptcy if applied in discrete time, if it involves shortselling or leveraging the risky asset.
 
  In order to prove Theorem \ref{thm:conv_stoch_dom_levy}, we therefore follow a different route. We first prove by induction the intermediate Proposition \ref{prop:conv_stoch_dom_euler}, which shows that the stochastic order holds true for discrete-time Euler approximations of $\hat X^\agent{RA}=\stochexp{\hat\pi_\agent{RA} R}$ and $\hat X^\agent{RL}=\stochexp{\hat\pi_\agent{RL} R}$. Theorem \ref{thm:conv_stoch_dom_levy} is then established by showing that the stochastic dominance is preserved in the limit.
  
  \begin{proposition}\label{prop:conv_stoch_dom_euler}
    Let $(R_i)_i$ denote a sequence of i.i.d.\ random variables and let $\hat\pi_\agent{RL}$, $\hat\pi_\agent{RA}\in\BbbR$ satisfying $\mathrm{sgn}(\hat\pi_\agent{RL})=\mathrm{sgn}(\hat\pi_\agent{RA})$ and $|\hat\pi_\agent{RL}|\geq |\hat\pi_\agent{RA}|$. Then 
    \begin{equation*}
      \prod_{i=1}^N\left(1+\hat\pi_\agent{RA}(R_i-\E[R_1])\right)\conv\prod_{i=1}^N\left(1+\hat\pi_\agent{RL}(R_i-\E[R_1])\right) \quad\forall N\in\BbbN.
    \end{equation*}
  \end{proposition}
  \begin{proof} By induction on $N$. For $N=0$ the assertion is trivial.

    For the  induction step $N-1\to N$ we apply Lemma \ref{lem:inher_conv_stoch_dom} (using the independence of the $R_i$ and the induction hypothesis) to obtain
    \begin{align*}
      &\left(\prod_{i=1}^{N-1}(1+\hat\pi_\agent{RA}(R_i-\E[R_1]))\right)(1+\hat\pi_\agent{RA}(R_N-\E[R_1]))\\
	&\qquad\conv \left(\prod_{i=1}^{N-1}(1+\hat\pi_\agent{RL}(R_i-\E[R_1]))\right)(1+\hat\pi_\agent{RA}(R_N-\E[R_1])).
    \end{align*}
    As Lemma \ref{lem:stoch_dom_rescaling} implies $(1+\hat\pi_\agent{RA}(R_N-\E[R_1]))\conv (1+\hat\pi_\agent{RL}(R_N-\E[R_1]))$, applying Lemma \ref{lem:inher_conv_stoch_dom} once again proves the result.
  \end{proof}
  
  Now we are in the position to prove Theorem \ref{thm:conv_stoch_dom_levy}:
    
  \begin{proof}[Proof of Theorem \ref{thm:conv_stoch_dom_levy}]
    Since the optimal fraction for power utility is independent of the initial capital $x$, we set w.l.o.g.\ $x=1$. By \cite[Proposition 4.4]{temme11notes}, $|\hat\pi_\agent{RA}|\leq|\hat\pi_\agent{RL}|$ and $\hat\pi_\agent{RA}$, $\hat\pi_\agent{RL}$ are non-negative (non-positive) if $b:=\expval{R_1}$ is non-negative (non-positive). Thus, Proposition \ref{prop:conv_stoch_dom_euler} implies
    \begin{equation}\label{eq:thm_conv_stoch_dom_levy}
      \prod_{i=1}^N\left(1+\hat\pi_\agent{RA}\left(\Delta_i^NR-\frac{bT}{N}\right)\right)\conv\prod_{i=1}^N\left(1+\hat\pi_\agent{RL}\left(\Delta_i^NR-\frac{bT}{N}\right)\right),
    \end{equation}
    where $\Delta_i^NR:=R_\frac{iT}{N}-R_\frac{(i-1)T}{N}$ and $\expval{\Delta_i^NR}=bT/N$. The left- and right-hand side of \eqref{eq:thm_conv_stoch_dom_levy} are Euler approximations on an equidistant grid with mesh width $T/N$ of the SDEs
    \begin{equation*}
      d\bar X^i_t=\hat\pi_i\bar X^i_t \;d\bar R_t\qquad i=\agent{RA},\agent{RL},
    \end{equation*}
    where $\bar R_t=R_t-bt$. Since $R$ is square-integrable, \cite[Theorem A.2]{temme11notes} shows that the Euler approximations converge in $L^1$ to the respective stochastic exponentials. As stochastic dominance is preserved under $L^1$-convergence we find
    \begin{equation*}
      \mathcal{E}(\hat\pi_\agent{RA} \bar{R})_T\conv \mathcal{E}(\hat\pi_\agent{RL} \bar{R})_T.
    \end{equation*}
    Using $\mathcal{E}(\hat\pi_i \bar{R})_T=\stochexp{\hat\pi_i R}_T\exp(-\hat\pi_i bT)$ and $\exp((\hat\pi_\agent{RL}-\hat\pi_\agent{RA})bT)\geq1$, Lemma \ref{lem:stoch_dom_rescaling} further implies
    \begin{equation*}
      \stochexp{\hat\pi_\agent{RA} R}_T\conv \stochexp{\hat\pi_\agent{RL} R}_T-\left(\exp(\hat\pi_\agent{RL}bT)-\exp(\hat\pi_\agent{RA}bT)\right).
    \end{equation*}
    By $\E[\stochexp{\hat\pi_iR}_T]=\exp(\hat\pi_i bT)$ the claim is proved.
  \end{proof}


\subsection{Extension to Models with Conditionally Independent Increments}
One can also extend Theorem \ref{thm:conv_stoch_dom_levy} to somewhat more general models of the risky asset $S$. 

Indeed, the proof of Theorem \ref{thm:conv_stoch_dom_levy} and corresponding auxiliary results only use the independence of the increments of the L\'evy process, but not their identical distribution. Hence, one can prove the convex order result of Theorem \ref{thm:conv_stoch_dom_levy} along the same lines for processes $S_t=\mathcal{E}(R)_t$, where $R$ has independent (but not necessarily identically distributed) increments. In this market model, the optimal policy for power utility is to invest a time-dependent but deterministic fraction $\hat\pi_t$ in the risky asset; the corresponding optimal wealth processes is then given by $\mathcal{E}(\int_0^\cdot\hat\pi_s dR_s)$.
  
This in turn allows to extend Theorem \ref{thm:conv_stoch_dom_levy} also to models with \emph{conditionally independent increments} (cf.\ \cite[Chapter II.6]{jacod_shiryaev03limit} for more details). Loosely speaking, this means that the return process $R$ of the risky asset $S$ is assumed to have independent increments with respect to an augmented filtration $\mathcal{G}_t$, that is, conditional on some stochastic factor processes. If these extra state variables are independent of the process driving the returns of the risky asset, Kallsen and Muhle-Karbe \cite{kallsen.muhlekarbe.10} show that the optimal policy is the same, both relative to the original and to the augmented filtration. With respect to the latter, one is dealing with a process with independent returns, such that  Theorem \ref{thm:conv_stoch_dom_levy} holds true. Statement \eqref{WeakResult} for the original filtration then follows immediately from the law of iterated expectations.

\renewcommand{\theequation}{A.\arabic{equation}}

\appendix
  \renewcommand{\thesection}{\Alph{section}}
  \renewcommand{\thedefinition}{\Alph{section}.\arabic{definition}}
\section{Auxiliary Results on the Convex Order}

In this appendix, we state and prove two elementary results on the convex order, that are needed for the proof of Proposition 
\ref{prop:conv_stoch_dom_euler}.

  \begin{lemma}\label{lem:stoch_dom_rescaling}
	Let $X$ be a random variable and $a\geq 1$ a real number. Then 
	$$X\conv a X-(a-1)\E[X].$$
  \end{lemma}
    \begin{proof}
      By Lemma \ref{lem:convmon_mon}, we have to prove $$\expval{\left(X-K\right)^+}\leq \expval{\left(aX-(a-1)\expval{X}-K\right)^+}$$ for all $K\in\BbbR$. By centering, it is easily seen that it sufficies to show the result for random variables $X$ with $\E[X]=0$.
    
      Let $K<0$. Since $\E[X]=0$ and $\int_{-\infty}^K x\, \BbbP^X(dx)\leq 0$ imply $\int_K^\infty x\, \BbbP^X(dx)\geq 0$, we find
      \begin{align*}
	\E[\left(X-K\right)^+]&\leq \int_K^\infty (x-K)\;\BbbP^X(dx)+(a-1)\int_K^\infty x \;\BbbP^X(dx),\\
		  &\leq \int_{K/a}^\infty (ax-K)\;\BbbP^X(dx)=\expval{\left(aX-K\right)^+}.
      \end{align*}
      Next assume $K\geq 0$. As $\int_{K/a}^K (ax-K)\; \BbbP^X(dx)\geq 0$, we  conclude
      \begin{align*}
       \expval{\left(X-K\right)^+}&\leq \int_K^\infty (ax-K)\;\BbbP^X(dx)\\
       &\leq \int_\frac{K}{a}^\infty (ax-K)\;\BbbP^X(dx)=\expval{\left(aX-K\right)^+}.
      \end{align*}
    \end{proof}

    \begin{lemma}\label{lem:inher_conv_stoch_dom}
      Let $X\conv Y$ and let $Z$ be independent of $X$ and $Y$. Then 
      $$XZ\conv YZ.$$
    \end{lemma}
    \begin{proof}
    Let $c:\BbbR\to\BbbR$ be convex. Then the result easily follows from $\BbbP^{ZX}=\BbbP^Z\otimes\BbbP^X$, $\BbbP^{ZY}=\BbbP^Z\otimes\BbbP^Y$, and since the function $\widetilde c(x):=c(zx)$ is again convex for all fixed $z\in\BbbR$.
    \end{proof}

\bibliography{bibliography}
\bibliographystyle{plain}
\end{document}